\documentclass[12pt]{article}
\usepackage[margin=1in]{geometry}
\usepackage{graphicx}
\usepackage{style/qstyle}

\title{Quantum interaction can superactivate cheating under parallel repetition}

\author[1]{Archishna Bhattacharyya\footnote{abhat086@uottawa.ca}}
\author[2]{Laura Man{\v{c}}inska\footnote{mancinska@math.ku.dk}}
\author[2]{Yuming Zhao\footnote{yuming@math.ku.dk}}
\affil[1]{Faculty of Science, University of Ottawa}
\affil[2]{QMATH, Department of Mathematical Sciences, University of Copenhagen}

\date{}

\numberwithin{equation}{section}
\numberwithin{figure}{section}

\begin{document}

\maketitle

\begin{abstract} 
We study interactive multiprover games and many-round protocols in which the communication between the verifier and the provers is quantum. Parallel repetition is known to suppress soundness error of classical two-prover games arbitrarily close to zero, as shown by Raz (STOC '95), even with quantum strategies as shown by Yuen (ICALP '16), and Bavarian, Vidick and Yuen (STOC '17). Yet, we show that quantum communication can have the opposite, unexpected effect. Specifically, we exhibit a one-round quantum game and two-round \(\mathsf{QMIP}\) protocols whose local (unentangled) value is strictly less than one for a single instance, yet equals one under \(n\)-fold parallel repetition for every \(n \ge 2\). The key is that parallel repetition not only imposes additional winning conditions but also supplies additional quantum resources as entanglement in the exchanged states can be exploited jointly across copies.

Our multiround examples build on superactivation of zero-error capacities of quantum channels. To establish this connection, we introduce quantum games and interactive protocols which capture the one-shot zero-error classical and quantum capacities, both with and without entanglement assistance. Our constructions use quantum-state verification and a teleportation-based reduction from two rounds to one.

In contrast, when shared entanglement is allowed, we show that parallel repetition cannot increase the entangled value when at most $3$ messages are exchanged, consistent with classical results of Bellare, Impagliazzo and Naor (FOCS '97), and that of single-prover systems by Kitaev and Watrous (STOC '00). Combining this monotonicity with our quantum game–channel correspondence, we show that entanglement-assisted zero-error classical and quantum capacities cannot be superactivated. 
\end{abstract}

\tableofcontents

\section{Introduction}

The interaction between quantum resources is fundamentally distinct from its classical counterparts. For instance, two different quantum channels may individually transmit zero bits of information, i.e., have zero \emph{capacity}\footnote{Informally, the capacity of a channel is the maximum amount of information it can send per use.}. Yet, it is possible that the combined channel is able to transmit a strictly nontrivial amount of information — an effect known as \emph{superactivation} of quantum channel capacities \cite{SY08}. On the other hand, there exist nonlocal games \cite{CHTW04} that can exhibit a strict separation between its classical and quantum values — such as the CHSH game \cite{CHSH69}. Entanglement is one such quantum resource that underpins much of this nontrivial interaction. Quantum channels as well as nonlocal games are both central to studying classically-impossible correlations and their role in information-processing, computing and cryptography.

Typically, in a two-player game with one-round of communication, or more generally in an interactive proof system \cite{GMR89} with many rounds, denoted by $G$: a referee samples a pair of questions $(x, y)$ from some probability distribution $\mu,$ sends $x$ to Alice and $y$ to Bob. Alice and Bob cannot communicate during the game. Alice responds with answer $a$ and Bob responds with answer $b,$ and the referee checks if $V(x, y, a, b) = 1$ for some predicate $V.$ If so, the players win $G.$ When the players are classical, their maximum success probability is called the classical value, denoted by $\omega(G).$ Concerned with the study of correlations, the question of what happens to $\omega(G)$ when the referee plays $n$ independent instances of $G$ in parallel with the two players, known as the $n$-fold parallel repetition of $G$, denoted by $G^n$ arises. Raz's parallel repetition theorem \cite{Raz95,Raz98} answers this as follows: if $\omega(G) < 1$ in a single instance of playing $G$, then the players' probability of winning all $n$ instances of the repeated game $G^n$ is exponentially small in $n.$ The result is intuitive; to win $G^{n}$, the provers must win every one of the $n$ parallel instances, so the repeated game seems strictly harder than a single copy.\footnote{This should be contrasted with quantum hedging phenomena \cite{MW12, AMR13}, where the objective is typically to win at least one of several repetitions rather than all of them.} 

\paragraph{Our contributions} Intriguingly, if the interaction between the referee and the players is quantum, that is, the questions and/or answers exchanged between the referee and players are now quantum states, we show that parallel repetition and the intuition it captures fails in the strongest possible way. Precisely, for a \textit{quantum game} $G$, we show a counterexample in \cref{thm:foul-play} where 
\begin{equation} \label{eq:counter-phenom}
\omega(G)<1 \qquad \text{but} \qquad \omega\left(G^{n}\right)=1, \qquad \text{for all}~ n \geq 2.
\end{equation}
Our findings are not merely restricted to one-round games, but more generally, we find classes of quantum multiprover interactive ($\mathsf{QMIP}$) protocols with at least $2$ nontrivial rounds of interaction where the same superactivation phenomenon is exhibited by the local (unentangled) value of the protocol, as shown in \cref{thm:yipee}, \cref{cor:all-n-yay}, and \cref{thm:phenom-quant}. Such an extreme violation of parallel repetition for any game or protocol is not known to date. In fact, \cite{CJPPG15} shows that parallel repetition is \textit{only approximately} satisfied for rank one quantum games, a much weaker consequence in comparison.

To unravel superactivation in quantum games, we establish a one-to-one quantum game-channel correspondence in \cref{thm:ch-pr-cor} and \cref{thm:Game-channel-correspondence}. Our constructions exploit the superactivation of channel capacities by introducing quantum games and interactive protocols that capture the one-shot zero-error entanglement-assisted classical and quantum capacities. Our result is a quantisation of the previously known correspondence between classical channels and (classical) nonlocal games \cite{CLMW10}.

Finally, we show that superactivated cheating as captured in \cref{eq:counter-phenom} is not possible for the entangled value, $\omega^{\ast}(G)$ of many-round $\mathsf{QMIP}$ protocols (\cref{thm:mult-imposs}) with at most $3$-messages exchanged. We also show that this is true for one-round quantum games (\cref{thm:imposs-q-game}). Combined with our quantum game-channel correspondence, we ultimately establish that the entanglement-assisted zero-error classical and quantum capacities cannot be superactivated, in \cref{thm:C-imp-sup} and \cref{cor:Q-imp-sup}.

\paragraph{Consequences of our results} An explanation of superactivated cheating of the local value occurring distinctly in \emph{quantum} games or protocols is that parallel repetition does not just impose additional winning constraints on the provers; it also provides them with additional quantum resources. Although the requirement of winning every instance becomes more stringent, the provers simultaneously gain access to more quantum states, including the entanglement contained in those states. This additional entanglement can be exploited across repetitions and may enable strategies that are impossible in a single instance. Thus two copies of the game can be won perfectly even though a single copy cannot. In our constructions, the extra copy of the question state provides precisely the resource needed to turn an imperfect one-copy strategy into a perfect strategy for the repeated game. The same follows through to $3$ repetitions, and consequently, parallel repetition need not reduce the value at all for any $n \geq 2.$

Since the failure of parallel repetition demonstrated also holds for $\mathsf{QMIP}$ protocols with $2$ nontrivial rounds, this shows yet another concrete separation between the complexity classes $\mathsf{QMIP}$ and $\mathsf{QMIP}^{\ast}.$ Notably, our result showing the impossibility of superactivated cheating for the entangled value of $3$-message $\mathsf{QMIP}$ protocols (\cref{thm:mult-imposs}) is a cryptographically relevant model, capturing a \emph{security game}. There is evidence showing that it may not be possible to extend our impossibility result to a larger number of rounds; as it is shown in \cite{BQSY24} that a computationally bounded 4-message $\mathsf{QMIP}$ protocol already faces barriers to parallel repetition. Our result in contrast is information-theoretic. Interestingly, \cite{BQSY24} also shows that parallel repetition is \textit{respected} by $\mathsf{QIP}$ (single prover) protocols in contrast with our counterexample \cref{eq:counter-phenom}. Most recently, \cite{LV26} reaffirms this while reporting that perfect completeness can be achieved for fewer than $3$-messages.

\paragraph{Cryptographic relevance} Another pertinent application of parallel repetition is error reduction. Often, noise in a protocol amounts to an adversary possessing side information. In such cryptographic contexts, a failure or violation of parallel repetition in addition to not preserving soundness will lead to uncontrolled noise amplification. This puts at risk adversarial tasks dependent on concurrent execution in interactive proof systems \cite{BIN97,KW00}. The extent to which a dishonest individual can attack such cryptographic protocols correlating independent executions is an important security consideration. For instance, already due to quantum hedging \cite{MW12}, where only a subset of the parallel-repeated instances are required to be won rather than all of them, in contrast to our case; attacks have been discovered for quantum money \cite{MVW13}. Our result is a stronger failure of parallel repetition, due to the possibility of superactivated cheating in any ($n \geq 2$) number of rounds, and thus we expect stronger cryptographic consequences. In the case of parallel repetition for a nonlocal game $K$ (with classical questions and answers), breaking a given cryptosystem is reduced to the task of distinguishing between $\omega(K) = 1$ and $\omega(K) < \delta$ for a certain $K.$ Parallel repetition can be employed to argue that by taking $H = \Gamma^n$ for some game $\Gamma$ this task is as hard as that of distinguishing between $\omega(\Gamma) = 1$ and $\omega(\Gamma) < 1 - \varepsilon_1$ when $n \geq \mathrm{poly}(\varepsilon^{-1}_1, \log \delta^{-1})$. Our findings pose an obstacle to such gap amplification techniques enhancing the security of protocols involving quantum interaction.

\paragraph{Connection with quantum PCP conjectures} A deeper interpretation of this result lies in the \textit{quantum PCP conjecture}. Classically, the main application of Raz's parallel repetition theorem is to the areas of \emph{hardness of approximation} and \emph{probabilistically checkable proofs (PCP)}. More precisely, parallel repetition is a fundamental tool for amplifying soundness of games and protocols. The PCP theorem is a cornerstone of complexity theory, and closely connected to the fact that the complexity class $\mathsf{MIP}$ is equal to $\mathsf{NEXP}.$ The PCP theorem can be formulated in terms of two-player games: there exists an $\varepsilon > 0$ such that it is $\mathsf{NP}$-hard to approximate the classical value of a game $G$ within additive error of $\varepsilon.$ Parallel repetition then gives a blackbox way to amplify this inapproximability, which lets us conclude that: for any $\varepsilon_0 > 0$ it is $\mathsf{NP}$-hard to approximate the classical value of a game $G$, within an additive error of $1 - \varepsilon_0.$ Since the value of a game is a number between $0$ and $1$, this implies strong inapproximability for games. In this landscape of inapproximability, the existence of a quantum version of the PCP theorem, and its precise proveable form has withstood long-standing interest. There are two proposed versions: the local Hamiltonian version asks whether it is $\mathsf{QMA}$-hard to determine the minimum eigenvalue of a local Hamiltonian with a constant energy gap, while the games version asks if there is an efficient $\mathsf{MIP}^{\ast}$ protocol for $\mathsf{QMA}$. Yet, our finding as stated in \cref{eq:counter-phenom} completely removes the possibility of using parallel repetition to amplify soundness, even just for a single repetition of a quantum game. 

In the quantum setting, the study of inapproximability of games displays richer behaviour, once again due to entanglement. Namely, there exists an \emph{entangled} value of games, denoted by $\omega^{\ast}(G)$ for a game $G.$ The entangled value is the maximum success probability, when the players are allowed to share unbounded entanglement as part of their strategy in the form of an arbitrary preshared entangled state. Recent landmark results \cite{Slo19} and \cite{JNV+21} respectively show that both computing and approximating the entangled value of a nonlocal game is $\mathsf{RE}$-hard or undecidable. The failure of parallel repetition, an ingredient that has been instrumental in \cite{JNV+21}, for quantum games is quite unprecedented. Our results pose an immediate obstacle to gap amplification when the interaction is quantum between provers and verifiers even with unentangled strategies. Gap amplification already for either version of the quantum PCP theorem has its challenges \cite{NN24a}. Our results add to this repertoire capturing that quantum games admit separations between the local value and the entangled value in terms of scaling; making the landscape of gap amplification for general games and protocols more nuanced. We suspect that this could influence the quantum PCP conjectures or at least the path to their proof of existence. 

\paragraph{Technical overview: bridging proof systems and channel capacities} Our results emerge from the fact that the winning probability is superactivated with repetitions of the game leading to the failure of parallel repetition. Estimating the rate of transmission through a quantum channel, $\mc{N}$ over many parallel uses, denoted by $\mc{N}^{\otimes n}$ is qualitatively the same as playing $n$ independent instances of a game in parallel, $G^n$ and analysing its value. The exact connection is in terms of resources and how they interact with each other. Entanglement indeed is a cause for a separation between how classical and quantum resources interact. For instance, it is enough to know that the input and output of a classical channel are uncorrelated to ascertain that the channel has zero capacity. However, we do not know a concrete answer to this question in the quantum context. Similarly, this extends to the compounding of resources, where for classical channels, the capacity of two channels is simply the sum of their individual capacities. Whereas, for quantum channels, and especially the quantum capacity, there may be strict separations as evidenced by superactivation \cite{SY08}. Over the years this has been extended to other capacities, namely zero-error classical and quantum capacities both in the one-shot and asymptotic setting \cite{Dua09, CS12, SS15}, and the finite-length setting of the finite-error quantum capacity \cite{PBDW26}. The reason for such extreme nonadditivity arises from the fact that the state across $n$ independent uses of a channel can be arbitrarily entangled. 

In our constructions, superactivation in channels is exploited to capture the extreme failure of parallel repetition in quantum games and protocols. First, given a quantum channel $\mc{N}$, we introduce classes of $\mathsf{QMIP}$ protocols $G, H$ each with $2$ non trivial rounds of interaction that models zero-error classical and quantum communication (see \cref{def:cl-cap-2ro} and \cref{def:Hseq}). Then we establish a one-to-one mapping from the zero-error classical and quantum capacity to the classical value of $G$ and $H$, and from their entanglement-assisted counterparts to the respective entangled values (see \cref{thm:ch-pr-cor} and \cref{thm:Game-channel-correspondence}). Intuitively, perfectly winning in each case corresponds to achieving full capacity. Once this is established we show that playing the repeated game in parallel is equivalent to receiving twice the amount of information through the joint channel from two uses. At this point we introduce a direct sum construction (\cref{def:OTP-dir-sum}) for channels such that any possible correlations between two channel outputs is broken by the direct sum of the two channels, which we call the \emph{one time pad} direct sum. This enables us then to establish that a failure of perfect execution in one instance is mitigated by superactivation of the same resources when interacting over multiple rounds, leading to the failure of parallel repetition. Furthermore, superactivated cheating for the local value of a one-round quantum game is established by considering the same principle for perfect state conversion under local operations and shared randomness.

Finally, we also explore the limitation of exploiting non trivial interaction among quantum resources to constrain entanglement. Already, whether the possibility of sharing unbounded entanglement allows the players to cheat under parallel repetition of the entangled value of a (classical) game is ruled out by Yuen's quantum parallel repetition theorem \cite{Yuen16}. We show that it is impossible for the entangled value of a single-round quantum game, and that of a $\mathsf{QMIP}$ protocol with at most $3$-messages exchanged ($1.5$ rounds) to superactivate (\cref{thm:mult-imposs}). This respects the intuition that when provers already have unbounded entanglement to share, they cannot use additional entanglement from question states to enhance the winning probability in concurrent instances. Furthermore, via teleportation, we convert our channel to game correspondence from multiple rounds to a single round (\cref{def:game-clas}, \cref{thm:game-prot-val}, and \cref{def:game-qu}, \cref{thm:ent-val-quant}), and using this fact show that it is impossible to superactivate any channel capacity asymptotically, if its one-shot version cannot be superactivated (\cref{thm:1-shot-superactivation-to-asymptotic}). As an application we show the impossibility of superactivating entanglement-assisted zero-error classical capacities (\cref{thm:C-imp-sup}), and the asymptotic case for its quantum analogue (\cref{cor:Q-imp-sup}), in stark contrast to their unassisted counterparts \cite{CS12}, of which no proof was known to exist to the best of our knowledge. 

\paragraph{Future directions}  In the context of soundness amplification or error reduction, one may ask: how can we circumvent this extreme failure of parallel repetition and still construct \emph{safe} quantum games? For classical games, the \emph{anchoring} transformation \cite{BVY17} achieves this, compensating for the lack of general strong parallel repetition\footnote{Recently, an exponential parallel repetition theorem for the entangled value of nonlocal (classical) games was put forward by Open AI in an AI-formulated document with other concurrent results, not available on a scientific preprint server, or peer-reviewed.}. Therefore, it is natural to ask if there is an analogue for quantum games. A trivial criterion to prevent cheating via superactivation is to ensure that the question state is separable. However, any non trivial criterion would be of immediate interest. Finally, whether superactivation of $\omega^*(G)$ is possible for protocols with a large number of rounds, or at least beyond $3$-messages is another interesting question.

\paragraph{Acknowledgements} AB thanks Anne Broadbent, Richard Cleve, Eric Culf, Debbie Leung, William Slofstra and Graeme Smith for helpful discussions on nonlocal games, channel capacities and cryptographic implications of parallel repetition. Part of this work was conducted during a research visit by AB at Institute for Quantum Computing, Waterloo; and she thanks Debbie Leung for hosting her. The authors thank Institut Mittag-Leffler and the programme, Operator Algebras and Quantum Information 2026 for hospitality during which some of the early discussions conceiving this work took place. AB acknowledges the support of the Natural Sciences and Engineering Research Council of Canada (NSERC)(ALLRP-578455-2022, RGPIN-2022-05167), the Air Force Office of Scientific Research under award number FA9550-20-1-0375, the Canada Research Chairs Program (CRC-2023-00173), and a QUORUM Travel Grant. LM and YZ are supported by an ERC grant (QInteract, Grant No 101078107).  

{\bf AI use.} We used an LLM to identify probability vectors with the properties required for the proof of \cref{thm:foul-play}.

\section{Preliminaries}

\paragraph{Notation} For $n\in\N$, write $[n]=\{1,2,\ldots,n\}$. For a function $f:\N\rightarrow[0,1]$, write $f:\N\rightarrow(0,1)_{\exp}$ to mean that there exist $N, k>0$ such that $2^{-n^k}<f(n)<1-2^{-n^k}$ for all $n\geq N$. We write $\log$ for the base-$2$ logarithm. We denote registers by uppercase Latin letters $A,B,C,\ldots$; and we denote Hilbert spaces by uppercase script letters $\mc{H},\mc{K},\mc{L},\ldots$. We always assume registers are finite sets and Hilbert spaces are finite-dimensional. We denote an independent copy of a register $A$ by $A'$. Given a register $A$, the Hilbert space spanned by $A$ is $\mc{H}_A=\spn\!\!\set*{\ket{a}}{a\in A}\cong\C^{|A|}$. We indicate that an operator or vector is on register $A$ with a subscript $A$, omitting when clear from context. Given two registers $A$ and $B$, we write $AB$ for their cartesian product, and treat the isomorphism $\mc{H}_{AB}\cong\mc{H}_A\otimes\mc{H}_B$ implicitly. Given finite-dimensional Hilbert spaces $\mc{H}$ and $\mc{K}$, we write $B(\mc{H},\mc{K})$ for the set of all linear operators $\mc{H}\rightarrow\mc{K}$, $B(\mc{H})=B(\mc{H},\mc{H})$, $\mc{U}(\mc{H})\subseteq B(\mc{H})$ for the subset of unitary operators, and $D(\mc{H})\subseteq B(\mc{H})$ for the subset of density operators where $D(\mc{H}) \coloneqq \{\rho \in B(\mc{H}), \rho \geq 0, \Tr{\rho} = 1\}$. 

An operator $A \in B(\mc{H})$ is positive, denoted by $A \geq 0$, if $A = \sqrt{A^{\dagger}A}$, where $(\cdot)^{\dagger}$ represents the Hermitian conjugate. We denote by $\text{id}_R$ the identity map on $B(\mc{H}_R)$. We write $\Tr$ for the trace on $B(\mc{H})$. On $B(\mc{H}_{AB})$, we write the partial trace $\Tr_{B}=\id\otimes\Tr$. For $\rho_{AB}\in B(\mc{H}_{AB})$, write $\rho_A=\Tr_B(\rho_{AB})$. We denote the $L_1$-norm by $\norm{\cdot}_1$ and the trace norm by $\frac{1}{2}\norm{\cdot}_1$. We denote the operator norm by $\norm{\cdot}$. The identity operator is denoted by $I$. We denote the canonical maximally-entangled state $\ket{\phi^+}_{AA'}=\frac{1}{\sqrt{|A|}}\sum_{a\in A}\ket{a}\otimes\ket{a}\in\mc{H}_{AA'}$. We write the maximally-mixed state on a register $A$ as $\omega_A=\frac{1}{|A|}\sum_{a\in A}\ketbra{a}\in D(\mc{H}_A)$. A positive-operator-valued measurement (POVM) is a finite set of positive operators $\{P_i\}_{i\in I}$ such that $\sum_iP_i=\mds{1}$, and a projection-valued measurement (PVM) is a POVM where all the elements are projectors. For a complex-valued random variable $X$, we write its expectation as $\expec X=\expec_XX$, and its variance as $\varsigma_X^2=\expec |X|^2-\abs*{\expec X}^2$.

A quantum channel is a completely positive trace-preserving (CPTP) map $\Phi:B(\mc{H})\rightarrow B(\mc{K})$. We denote the Choi-Jamio\l{}kowski isomorphism $J:B(B(\mc{H}_A),B(\mc{H}_B))\rightarrow B(\mc{H}_{AB})$, $J(\Phi)=(\id\otimes\Phi)(\ketbra{\phi^+}_{AA'})$. Note that if is $\Phi$ is a quantum channel, $J(\Phi)\in D(\mc{H}_{AB})$, called the Choi state. The Kraus representation of a quantum channel $\Phi: B(\mc{H}) \to B(\mc{H})$ is $\sum \limits_{i = 0}^{d - 1} \Phi(\rho) = A_i \rho A_i^{\dagger}$ where $A_i \in B(\mc{H})$ are the Kraus operators such that $\sum \limits_{i = 0 }^{d - 1} A_i^{\dagger} A_i = I,$ and $d = \text{dim}(\mc{H})$. For $\Phi: B(\mc{H}_A) \to B(\mc{H}_B)$, by Stinespring's dilation theorem, there exists an isometry $V: A \to BE$ known as the Stinespring isometry such that $\Phi (\rho) = \Tr_E (V \rho V^{\dagger})$. The complementary channel $\Phi^c: B(\mc{H}_A) \to B(\mc{H}_E)$ is given by $\Phi^c (\rho) = \Tr_B (V \rho V^{\dagger}).$ Note that the Kraus representation is not unique. However, for a channel $\Phi$ we can always find its minimal Stinespring representation and corresponding minimal Kraus representation. Every Stinespring representation is equivalent to a minimal one via an isometry on the environment denoted by register $E$.

\subsection{Quantum games and interactive protocols}

A (single-round) quantum two-party game is specified by a tripartite question state
$\ket{\psi}_{ABR} \in \mathcal{H}_A \otimes \mathcal{H}_B \otimes \mathcal{H}_R,$
and an accepting projection
    \[
        \Pi_{\mathrm{acc}} \in \mathrm{Proj}\bigl(
        \mathcal{H}_a \otimes \mathcal{H}_b \otimes \mathcal{H}_R
        \bigr).
    \]
Note that the question spaces $\mathcal{H}_A$ and $\mathcal{H}_B$ need not be the same as the answer spaces $\mathcal{H}_a$ and $\mathcal{H}_b$.

The game is played as follows:
\begin{enumerate}
    \item The referee prepares the question state $\ket{\psi}_{ABR}$, sends register $A$ to Alice and register $B$ to Bob, and keeps register $R$.
    
    \item Alice and Bob apply their respective strategies to the registers they receive, returning the outputs of local quantum channels
\[
    \mathcal{E} : \mathrm{L}(\mathcal{H}_{A'}) \to \mathrm{L}(\mathcal{H}_a),
    \qquad
    \mathcal{F} : \mathrm{L}(\mathcal{H}_{B'}) \to \mathrm{L}(\mathcal{H}_b),
\]
where
$\mathcal{H}_{A'} = \mathcal{H}_A$ and
$\mathcal{H}_{B'} = \mathcal{H}_B$ if shared entanglement is not allowed. If it is allowed, then
$\mathcal{H}_{A'} = \mathcal{H}_A \otimes \mathcal{H}_{E_A}$ and
$\mathcal{H}_{B'} = \mathcal{H}_B \otimes \mathcal{H}_{E_B}$, where
$\mathcal{H}_{E_A}$ and $\mathcal{H}_{E_B}$ contain Alice's and Bob's shares
of the entangled state.
    
    \item Upon receiving the answer registers, the verifier performs the two-outcome measurement
    \[
        \bigl\{ \Pi_{\mathrm{acc}},\ I - \Pi_{\mathrm{acc}} \bigr\}
    \]
    on the joint system $abR$. Alice and Bob win if the outcome corresponding to $\Pi_{\mathrm{acc}}$ occurs.
\end{enumerate}

A strategy $S$ for $G=(\ket{\psi},\Pi_{\text{acc}})$ specifies the local operations performed by Alice and Bob. A local strategy $S=(\mc{E},\mc{F})$ consists of quantum channels $\mc{E}:L(\mc{H}_{A'})\rightarrow L(\mc{H}_a)$ and $\mc{F}:L(\mc{H}_{B'})\rightarrow L(\mc{H}_b)$. Its winning probability is given by
\begin{equation*}
    \omega(G;S)= \Tr \big(\Pi_{\text{acc}}(\mc{E}\otimes\mc{F}\otimes\id_R) (\ketbra{\psi})   \big)
\end{equation*}
An entanglement-assisted strategy $S=(\mc{E},\mc{F},\ket{\varphi}_{E_AE_B})$ also includes a shared state $\ket{\varphi}_{E_AE_B}$. The winning probability is given by
\begin{equation*}
    \omega(G;S)= \Tr \big(\Pi_{\text{acc}}(\mc{E}\otimes\mc{F}\otimes\id_R) (\ketbra{\psi}\otimes\ketbra{\varphi})\big).
\end{equation*}
The local value and entanglement-assisted value of $G$ are then defined by
\begin{equation*}
    \omega(G):=\sup\{\omega(G;S): S \text{ is a local strategy} \}
\end{equation*}
and
\begin{equation*}
    \omega^*(G):=\sup\{\omega(G;S): S \text{ is an entanglement-assisted strategy} \}.
\end{equation*}

For a fixed quantum game, an unassisted strategy consists of two channels between fixed finite-dimensional spaces. Using their Choi matrices, the set of such channels is compact, and $S\mapsto \omega(G;S)$ is continuous. Therefore $\omega(G)$ is always attained by some strategy, and one can replace $\sup$ with $\max$. However, the analogous conclusion need not hold for the entanglement-assisted value, because the players may share an entangled state of arbitrarily large dimension. In fact, there exists a nonlocal game with quantum value $1$ but does not have a perfect finite-dimensional strategy~\cite{Slo19}. Viewing this game as a quantum game, we conclude that the supremum in the definition of $\omega^*(G)$ cannot in general be replaced by a maximum.

A quantum multiprover interactive proof ($\mathsf{QMIP}$) is a generalisation of an interactive quantum game to multiple rounds with multiple provers. A round of communication consists of two turns of exchange between the verifier and provers. Without loss of generality, all of the verifier and prover actions in a multi-prover interactive protocol are represented by isometric channels acting on pure states, or unitary channels provided that sufficiently large ancillary spaces are made available for each participant at the start of the game. The initial state of the $k$ provers’ private registers $P_0^1, \ldots, P_0^k)$ play a particularly important role in $\mathsf{QMIP}$ protocols as they might benefit from shared starting states (especially entangled states) that cannot be prepared locally. Provers whose private registers are initialised to a product state, which could be prepared locally and independently by each prover, are referred to as \emph{unentangled} provers. General provers, on the other hand, are permitted to start the game with the collection of registers $(P_0^1, \ldots, P_0^k)$ initialised in an arbitrary quantum state. Such provers are typically be called \emph{entangled} provers, and the shared starting state is referred to as their prior shared entanglement. The unentangled or local value of a $\mathsf{QMIP}$ protocol, $G$, denoted by $\omega(G)$ is the highest probability with which the verifier can be made to output $1$ when interacting with provers whose private registers are all initialised to the all zero product state $\ket{0 \cdots 0}.$ The entangled value $\omega^{\ast}(G)$ is defined as the supremum over all finite-dimensional Hilbert spaces $\mc{H}_0^1, \ldots, \mc{H}_0^k$ corresponding to the provers’ initial private registers $P_0^1, \ldots, P_0^k$ and all initial pure states $\ket{\psi} \in \mc{H}_0^1 \otimes \cdots \otimes \mc{H}_0^k$ of these registers, of the provers’ maximum probability of causing the verifier to output $1.$ These two values lead to potentially distinct classes of problems having multi-prover interactive proof systems: $\mathsf{QMIP}$ for the case of
unentangled provers and $\mathsf{QMIP}^{\ast}$ when the provers are allowed to share arbitrary entangled states.

\paragraph{Parallel repetition.} Given two games $G=(\ket{\psi},\Pi_{\text{acc}})$ and $G'=(\ket{\psi'},\Pi'_{\text{acc}})$ their parallel repetition is a quantum game $G\times G' := (\ket{\psi}\otimes\ket{\psi'},\Pi_{\text{acc}}\otimes\Pi'_{\text{acc}})$. We can, of course, also take $G=G'$ as well as consider $k$-fold parallel repetition.

For parallel repetition with a two-prover, multi-round $\mathsf{QMIP}$ protocol, the verifier acts as $n$ independent copies of that round for one instance. Most generally, the verifier's action can be modelled as a quantum channel, $\mc{M}$ say, and under $n$ instances of parallel repetition, he will act as $\mc{M}^{\otimes n}.$

\subsection{Quantum channel capacities}

A quantum channel can send both classical and quantum data. Informally, the maximum number of classical or quantum bits of information over the number channel use defines the quantum or classical capacity of a quantum channel. Depending on the type of data sent, the resources used in transmission and whether the transmission is perfect, there exist many varieties of capacities for quantum channels. Let $\mc{N}$ be a quantum channel and denote by $X(\mc{N})$ a general capacity of $\mc{N}.$ For simplicity, we consider sending qubits, but it holds more generally for arbitrary dimensional channels.

\begin{definition}
    $X(\mc{N}) \coloneqq \lim\limits_{\epsilon \to 0} \lim\limits_{n \to \infty} \sup \{ \frac{m}{n} : \exists ~{\mc{E}}~ \exists~ \mc{D}~ \forall~ {\psi \in \Gamma_m}, F(\mc{D} \circ \mc{N}^{\otimes n} \circ \mc{E}(\psi), \psi) > 1 - \epsilon \},$
    where $\mc{E}$ is an encoding subprotocol to be performed by the sender, Alice who receives an $m$-qubit state $\psi$ belonging to some set $\Gamma_m$ of permissible inputs to the entire protocol, and produces $n$ possibly entangled inputs to the channel $\mc{N}$; $\mc{D}$ is a decoding subprotocol performed by the receiver, Bob who receives $n$ possibly entangled channel outputs and produces an $m$-qubit output for the entire protocol. $F(\cdot, \cdot)$ denotes the fidelity of the output relative to the input $\psi$, i.e., the probability that the output state would pass a test determining whether it is equal to the input. 
\end{definition}
In the above, when $\epsilon = 0$ this corresponds to perfect transmission and introduces the notion of zero-error capacities. Also, the notion so far is defined in the asymptotic limit, $n \to \infty$. One can also consider the supremum over all achievable rates for a single use of the channel, and that refers to the one-shot capacity, denoted by $X^{(1)}(\mc{N})$. If Alice and Bob preshare an entangled state prior to running their subprotocols, this corresponds to the notion of entanglement-assisted capacities. Often, it so happens that an entropic quantity characterises the one-shot capacity of a channel, and to obtain the asymptotic capacity, a regularisation is needed. More clearly, $X(\mc{N}) \coloneqq \lim\limits_{n \to \infty} \frac{X^{(1)}(\mc{N}^{\otimes n})}{n}.$ For instance, the quantum capacity of a channel exhibits this behaviour as established by the quantum capacity theorem due to \cite{Lloyd97,Shor02,Devetak05}.

Certain channel capacities also display a peculiar behaviour known as superactivation: there exist channels $\mc{N}_1, \mc{N}_2$ such that a certain capacity $X(\mc{N}_1) = X(\mc{N}_2) = 0$ but $X(\mc{N}_1 \otimes \mc{N}_2) > 0.$ This was first discovered for the quantum capacity by Smith and Yard \cite{SY08} and then extended to various other cases, some of which are essential for our main results as will be observed in later sections. Two capacities, the one-shot entanglement-assisted zero-error classical and quantum capacities are particularly useful for our analysis denoted as $C_{0,E}^{(1)}(\mc{N})$ and $Q_{0,E}^{(1)}(\mc{N})$, and we define them as follows.

\begin{definition}[Entanglement-assisted one-shot zero-error classical capacity]
    The entanglement-assisted one-shot zero-error classical capacity of a CPTP map $\Phi: B(\mc{H}_A) \to B(\mc{H}_B)$ is defined as
\begin{equation}
    C^{(1)}_{0,E}(\Phi) \coloneqq \sup \limits_{\ket{\psi}, M} \log \lvert M \rvert, 
\end{equation}
where the supremum is over all pure bipartite (entangled) states $\ket{\psi} \in D(\mc{H}_{{A_0}} \otimes \mc{H}_{{B_0}})$ and collections $M$ of CPTP maps $\{\mc{E}_m : B(\mc{H}_{A_0}) \to B(\mc{H}_A)\}_{m = 1}^{|M|}$ such that $$\forall~ m \neq m' : (\Phi \circ \mc{E}_m \otimes \text{id}_{B_0}) (\ket{\psi}) \perp (\Phi \circ \mc{E}_{m'} \otimes \text{id}_{B_0}) (\ket{\psi}).$$ 
\end{definition}

\begin{definition}[Entanglement-assisted one-shot zero-error quantum capacity]
    The entanglement-assisted one-shot zero-error quantum capacity of a CPTP map $\Phi: B(\mc{H}_A) \to B(\mc{H}_B)$ is defined as
\begin{equation}
    Q_{0,E}^{(1)}(\Phi) \coloneqq \sup\limits_{\ket{\psi}, ~\mc{C}} \log \text{dim}(\mc{C}),
\end{equation}
where the supremum is over all pure bipartite (entangled) states $\ket{\psi} \in D(\mc{H}_{A_0} \otimes \mc{H}_{B_0})$, and CPTP maps $\mc{E} : B(\mc{C} \otimes \mc{H}_{A_0}) \to B(\mc{H}_A)$, for all subspaces $\mc{C}$ for which there exists a recovery channel $\mc{R}: B(\mc{H}_B \otimes \mc{H}_{B_0}) \to B(\mc{C})$ satisfying $\mc{R}\Big(((\Phi \circ \mc{E}) \otimes \text{id})(\rho \otimes \psi)\Big) = \rho$ for all $\rho \in D(\mc{C})$ with supp$(\rho) \subseteq \mc{C}$.
\end{definition}

In the absence of the shared entangled state between the sender and receiver, the above capacities reduce to the usual unassisted classical and quantum zero-error capacities.

\subsection{SWAP test}

The well-known SWAP test \cite{BCWdW01} was introduced in the context of quantum fingerprinting for a verifier to test if two quantum states are identical $\ket{\psi} = \ket{\phi}$, or have inner product at most $\delta$, i.e., $\abs*{\braket{\phi}{\psi}} \leq \delta.$ This is a purely quantum operation with no classical analogue. It works as follows: the test is a procedure that measures and outputs the first qubit of the state $(H \otimes I)(\text{c-SWAP})(H \otimes I)\ket{0}\ket{\phi}\ket{\psi};$ where $H$ is the Hadamard transform mapping $\ket{b} \to \frac{1}{\sqrt{2}}(\ket{0} + (-1)^b\ket{1}),$ SWAP is the operation $\ket{\phi}\ket{\psi} \to \ket{\psi}\ket{\phi},$ and c-SWAP is the controlled SWAP gate —~controlled by the first qubit. Through the execution of this operation, one gets that the final state before the measurement is $\frac{1}{2}\ket{0}(\ket{\phi}\ket{\psi} + \ket{\psi}\ket{\phi}) + \frac{1}{2}\ket{1}(\ket{\phi}\ket{\psi} - \ket{\psi}\ket{\phi}).$ Measuring the first qubit of this state produces outcome $1$ with probability $\frac{1}{2} - \frac{1}{2}\abs*{\braket{\phi}{\psi}}^2.$

\section{Superactivation of cheating for quantum games}

In this section we show, for one-round quantum games, it is possible that the classical value of the games superactivates under parallel repetition, i.e., the entanglement in the question state allows the players to cheat or win perfectly in one or more repetitions even when it is impossible for them to perfectly in a single instance (\cref{thm:foul-play}). However, we show that the opposite is true for the entangled value of quantum games, as expected. That is, the entangled value of two instances is always bounded by that of a single instance (\cref{thm:imposs-q-game}).

\begin{proposition} \label{thm:imposs-q-game}
For any quantum games $G_1,G_2$, we have $\omega^*(G_1\times G_2) \le \omega^*(G_i)$, where $i=1,2$. 
\end{proposition}

\begin{proof}
We prove the claim for $i=1$; the proof for $i=2$ is identical. Let
\[
    (\varphi_{E_AE_B}, \mathcal{E}, \mathcal{F})
\]
be an arbitrary entanglement-assisted strategy for $G_1 \times G_2$, with
success probability $p$. Thus
\[
    \mathcal{E} :
    \mathrm{L}(\mathcal{H}_{A_1} \otimes \mathcal{H}_{A_2}
    \otimes \mathcal{H}_{E_A})
    \to
    \mathrm{L}(\mathcal{H}_{a_1} \otimes \mathcal{H}_{a_2})
\]
and similarly
\[
    \mathcal{F} :
    \mathrm{L}(\mathcal{H}_{B_1} \otimes \mathcal{H}_{B_2}
    \otimes \mathcal{H}_{E_B})
    \to
    \mathrm{L}(\mathcal{H}_{b_1} \otimes \mathcal{H}_{b_2}).
\]

We construct an entanglement-assisted strategy for $G_1$ that succeeds with probability at least $p$. In addition to
$\varphi_{E_AE_B}$, Alice and Bob share the reduced question state of the
second game,
\[
    \rho_{A_2B_2}
    =
    \operatorname{Tr}_{R_2}
    \bigl(\ket{\psi_2}\bra{\psi_2}_{A_2B_2R_2}\bigr).
\]
Upon receiving $A_1$ and
$B_1$, they apply the same channels $\mathcal{E}$ and $\mathcal{F}$ as in the
strategy for $G_1 \times G_2$, and then return only the registers $a_1$ and
$b_1$, discarding $a_2$ and $b_2$.

Let $\rho_{a_1a_2b_1b_2R_1R_2}$ be the state produced by the original strategy
for $G_1 \times G_2$. Its success probability is
\[
    p
    =
    \operatorname{Tr}\Bigl[
        \bigl(\Pi_{\mathrm{acc}}^{(1)} \otimes
        \Pi_{\mathrm{acc}}^{(2)}\bigr)
        \rho_{a_1a_2b_1b_2R_1R_2}
    \Bigr].
\]
Since $\Pi_{\mathrm{acc}}^{(1)} \otimes \Pi_{\mathrm{acc}}^{(2)}
    \preceq
\Pi_{\mathrm{acc}}^{(1)} \otimes I_{a_2b_2R_2}$,
the strategy constructed for $G_1$ succeeds with probability
\[
    \operatorname{Tr}\Bigl[
        \bigl(\Pi_{\mathrm{acc}}^{(1)} \otimes I_{a_2b_2R_2}\bigr)
        \rho_{a_1a_2b_1b_2R_1R_2}
    \Bigr]
    \geq p.
\]
Therefore every entanglement-assisted strategy for $G_1 \times G_2$ gives an
entanglement-assisted strategy for $G_1$ with success probability at least as
large. Taking the supremum over all strategies for $G_1 \times G_2$ gives the desired statement.
\end{proof}

\begin{corollary}
    For a quantum game $G$, if $\omega(G^2)=p$ then $\omega^*(G)\ge p$.
\end{corollary}

\begin{remark}
    For two distinct games the above statements are false. The following counterexample shows that we can find games, where $\omega(G\times H)>\omega(G)$. Consider a game $G$ with $\omega^*(G)>\omega(G)$. Let $H$ be a game where the accepting projection is identity, i.e., the players always win but the question state of $H$ supplies the players with the entangled state needed to reach $\omega^*(G)$. Then $\omega(G\times H) = \omega^*(G) >\omega(G)$.
\end{remark}

We now show our main result, that the following phenomenon is true for quantum games.

\begin{theorem} \label{thm:foul-play}
There exists a quantum game $G$ such that $\omega(G)<1$ but $\omega(G^{n}) = 1$ for all $n \geq 2$.
\end{theorem}

To prove this, we first define the quantum game we work with, denoted by $G_{\mathrm{SC}}$. It captures the information-processing task of state conversion.

\begin{definition}
    Let $G_{\mathrm{SC}}$ be a quantum game with $R, A, B$ denoting the referee, and players, Alice and Bob's systems respectively. Fix $p$ and $q$ to be probability vectors. Then the game is played as follows:
    \begin{enumerate}
        \item Let $R$ be trivial, and let $R$ send $A$ and $B$ the bipartite pure state $\ket{\psi_p}_{AB} = \sum \limits_i \sqrt{p_i} \ket{i, i}.$
        \item $R$ accepts if and only if the returned state is $\ket{\phi_q}_{ab} = \sum_j \sqrt{q_j} \ket{j, j},$ i.e., the accepting projection $\Pi_{acc} = \ketbra{\phi_q}.$
    \end{enumerate}
    The local strategy needed to play this game is:
    \begin{enumerate}
        \item A pair of local channels $\mc{E}$ and $\mc{F}$.
        \item The winning probability corresponding to a local strategy is given by $\omega(G_{\mathrm{SC}}) \coloneqq \braket{\phi_q}{(\mc{E} \otimes \mc{F})(\psi_p)}{\phi_q}.$
        \item A perfect local strategy would then exactly correspond to a pair of channels $\mc{E}$ and $\mc{F}$ such that $(\mc{E} \otimes \mc{F}) (\psi_p) = \phi_q.$
        \item A perfect local strategy corresponds to perfect state conversion under local operations and shared randomness (LOSR).
    \end{enumerate}
\end{definition}

Now we prove our main result of this section. 

\begin{proof}[Proof of \cref{thm:foul-play}]
    Let $G = G_{\mathrm{SC}}.$ We shall prove this by showing that there exist probability vectors $p, q$ such that perfect state conversion under LOSR fails for one copy of the state, but succeeds for $2$ and $3$ copies of the resource state, and hence $n$. A simple criterion for exact conversion of pure state $\psi_p$ to $\phi_q$ is if and only if $\exists~ r$, a probability vector, such that $p = q \otimes r$, up to a permutation. To prove our claim, we want that $p \neq q \otimes s$, for every probability vector $s$, while $p^{\otimes 2} = q^{\otimes 2} \otimes r_2$ and $p^{\otimes 3} = q^{\otimes 3} \otimes r_3$ for some probability vectors $r_2, r_3$. 
    
    Consider the polynomial $S(x) = 1+x+2x^2-x^3+2x^4+x^5$ and define $P(x) =(1+x)S(x) = 1+2x+3x^2+x^3+x^4+3x^5+x^6.$ Although $S(x)$ has a negative coefficient, both its square and cube have nonnegative coefficients: $S(x)^2 = 1+2x+5x^2+2x^3+6x^4+2x^5 +11x^6+2x^8+4x^9+x^{10}$ and $S(x)^3 = 1+3x+9x^2+10x^3+18x^4+12x^5+35x^6+18x^7 +36x^8+5x^9+33x^{10}+18x^{11} +2x^{12}+9x^{13}+6x^{14}+x^{15}.$
    Set $x=2$ and define $q = \left(\frac{1}{3},\frac{2}{3}\right)$. Let $p$ be the normalised multiset $p = \frac{1}{201}\bigl(1,2,2,4,4,4,8,16,32,32,32,64\bigr).$
    
    The multiplicities of the dyadic entries of $p$ are encoded by $P(x)$. Since $P(x)^2 = (1+x)^2S(x)^2$ and $S(x)^2$ has nonnegative integer coefficients, there exists a probability vector $r_2$ such that, up to permutation, $p^{\otimes 2} = q^{\otimes 2}\otimes r_2.$ Likewise, $P(x)^3 = (1+x)^3S(x)^3,$ and the nonnegativity of the coefficients of $S(x)^3$ implies that there exists a probability vector $r_3$ satisfying $p^{\otimes 3} = q^{\otimes 3}\otimes r_3.$ Consequently, the corresponding quantum game $G$ satisfies $\omega\!\left(G^2\right) = \omega\!\left(G^3\right) = 1.$
    
    On the other hand, $p\neq q\otimes s$ for every probability vector $s$. Indeed, since $q=(1,2)/3$, such a factorization would require the entries of $p$ to be partitionable into pairs in the ratio $1:2$. The multiplicities of the powers $1,2,4,8,16,32,64$ in $201p$ are $(c_0,c_1,c_2,c_3,c_4,c_5,c_6) = (1,2,3,1,1,3,1).$ If $t_e$ denotes the number of pairs of the form $(2^e,2^{e+1})$, then we would need $c_e=t_e+t_{e-1}$, $t_{-1}=0.$ This recursively gives $t_0=1,$ $t_1=1,$ $t_2=2,$ $t_3=-1,$ which is impossible. Thus no exact one-copy conversion exists, and hence $\omega(G)<1.$
    
    In fact, this example satisfies the stronger conclusion $\omega(G)<1, ~\omega(G^n)=1 ~\text{for every }n\geq 2.$
    Indeed, every integer $n\geq 2$ can be written as $n=2a+3b$ for some $a,b\geq 0$. It follows that $S(x)^n = \bigl(S(x)^2\bigr)^a \bigl(S(x)^3\bigr)^b$ has nonnegative coefficients, and therefore $p^{\otimes n} = q^{\otimes n}\otimes r_n$ up to permutation for some probability vector $r_n$.
\end{proof}

\section{Interactive quantum protocols for zero-error capacities}

In this section we formalise quantum games and interactive protocols for zero-error classical and quantum capacities, and their entanglement-assisted counterparts. This is at the core of our construction of the counterexamples. \cref{thm:ch-pr-cor} and \cref{thm:Game-channel-correspondence} provide a mapping from channels to interactive protocols (and games subsequently) for zero-error classical and quantum capacities respectively.

\subsection{Zero-error classical capacities}

In this section we restrict our attention to classical capacities.

\begin{definition}[Sequential 2-round protocol] \label{def:cl-cap-2ro}
    Let $G_{seq}$ be a QMIP protocol where $V, A, B$ denote the systems of the verifier, and players Alice and Bob respectively. Let $m, m' \in \left[M\right]$ be input and output messages through a noisy quantum channel $\mc{N}.$ The game is played as follows:
    \begin{enumerate}
        \item $V$ samples $m \leftarrow \left[M\right]$ and sends it to $A$.
        \item $A$ then outputs $\rho_A$, her share of a quantum state $\rho_{AB}$ potentially entangled with $B$.
        \item $V$ applies $\mc{N}: B(\mc{H}_a) \to B(\mc{H}_b)$ to $\rho_A$ and obtains $\rho_{B'}$ which he then sends to $B'$ as a question.
        \item $B$ receives $\rho_{B'}$ and then sends $m' \in \left[M\right]$ as the answer to $V.$
        \item The players win if $m = m'.$
    \end{enumerate}

    The strategy needed to play this game is as follows:
    \begin{enumerate}
        \item $A$ and $B$ share an entangled state $\rho_{AB}.$ 
        \item $A$ has an encoding channel $\mc{E}_m : B(\mc{H}_A) \to B(\mc{H}_a)$ for each $m \in \left[M\right].$
        \item $B$ performs a decoding POVM $\{D_{m'} \in B(\mc{H}_B) \otimes B(\mc{H}_b)\}$ for each $m' \in \left[M\right].$
    \end{enumerate}
\end{definition}

\begin{theorem}[From channels to protocols] \label{thm:ch-pr-cor}
    Given a channel $\mc{N}$, there exists $\mathsf{QMIP}$ protocol $G_{seq}$ such that the following hold true:
    \begin{enumerate}
        \item $G_{seq}$ has a perfect entanglement-assisted strategy if and only if $C_{0,E}^{(1)}(\mc{N}) \geq \log \abs*{M}.$
        \item $\omega(G_{seq}) = 1$ if and only if $C_0^{(1)}(\mc{N}) \geq \log \abs*{M}.$
    \end{enumerate}
\end{theorem}

\begin{proof}
   We first prove the entanglement-assisted statement. Let
\begin{equation*}
    S=
    \left(
        \rho_{AB},
        \{\mc{E}_m\}_{m\in[M]},
        \{D_{m'}\}_{m'\in[M]}
    \right)
\end{equation*}
be an entanglement-assisted strategy for
$G_{seq}(\mc{N},M)$, and define
\begin{equation*}
    \sigma_m
    :=
    \bigl(
        (\mc{N}\circ\mc{E}_m)\otimes\operatorname{id}_B
    \bigr)(\rho_{AB}).
\end{equation*}
Its winning probability is
\begin{equation*}
    \omega(G_{seq};S)
    =
    \frac{1}{M}
    \sum_{m\in[M]}
    \Tr(D_m\sigma_m).
\end{equation*}

Suppose first that $S$ is perfect. Since $0\leq\Tr(D_m\sigma_m)\leq 1$ for every $m$, the equality $\omega(G_{\mathrm{seq}};S)=1$ implies that $\Tr(D_m\sigma_m)=1$ for every $m\in[M]$. As $\{D_{m'}\}_{m'\in[M]}$ is a POVM,we also have $\Tr(D_{m'}\sigma_m)=0$ whenever $m'\neq m$. Thus the states $\{\sigma_m:m\in[M]\}$ are perfectly distinguished by the POVM $\{D_m\}_{m\in[M]}$. It follows that $ C_{0,E}^{(1)}(\mc{N})\geq\log \abs{M}$

Conversely, suppose that $ C_{0,E}^{(1)}(\mc{N})\geq\log \abs{M}$. The we have a shared state $\rho_{AB}$, encoding channels $\{\mc{E}_m\}_{m\in[M]}$, and a decoding POVM $\{D_m\}_{m\in[M]}$ satisfying $\Tr(D_m\sigma_m)=1$ for every $m\in[M]$. These define an entanglement-assisted strategy $S$ for $G_{seq}$ with $\omega(G_{\mathrm{seq}};S)=1$.

The local statement follows by the same argument without a
shared entangled state: $G_{seq}$ admits a
perfect local strategy if and only if $C_0^{(1)}(\mc{N})\geq\log \abs{M}$.
The result follows from the fact that the existence of a perfect local strategy is equivalent to $\omega(G_{seq}) = 1$.
\end{proof}

\begin{definition}[Parallel 1-round game] \label{def:game-clas}
    Let $G_{par}$ be a quantum game where $V, A, B$ denote the systems of the verifier, and players Alice and Bob respectively. Let $m, m' \in \left[M\right]$ be input and output messages through a noisy quantum channel $\mc{N}$. The game is played as follows:
    \begin{enumerate}
        \item $V$ samples $m \leftarrow \left[M\right]$ and sends it to $A$, and prepares $\nu_{ab}$, the Choi state of $\mc{N}$ and sends part $b$ to $B.$
        \item $A$ sends $\rho_A$ to $V$ which is her share of the entangled state $\rho_{AB}.$
        \item $B$ measures the system $Bb$ and sends the outcome $m'$ to $V.$
        \item $V$ measures in the Bell basis of $Aa$. If the outcome is not $\ket{\phi^{+}}$ then $V$ accepts. If the outcome is $\ket{\phi^{+}}$ then $V$ accepts if $m = m'.$
        \item The players win when $V$ accepts.
    \end{enumerate}

    The strategy needed to play this game is as follows:
    \begin{enumerate}
        \item $A$ and $B$ share an entangled state $\rho_{AB}.$ 
        \item $A$ has an encoding channel $\mc{E}_m : B(\mc{H}_A) \to B(\mc{H}_a)$ for each $m \in \left[M\right].$
        \item $B$ performs a decoding POVM $\{D_{m'} \in B(\mc{H}_B) \otimes B(\mc{H}_b)\}$ for each $m' \in \left[M\right].$
    \end{enumerate}
\end{definition}

\begin{lemma}[Multi-to-single round equivalence]\label{thm:game-prot-val}
    Given a channel $\mc{N}$, there exists a quantum game $G_{par}$ such that the following hold true:
    \begin{enumerate}
        \item $1 - \frac{1 - \omega^{\ast}(G_{seq})}{d^2_a} = \omega^{\ast}(G_{par}).$
        \item $1 - \frac{1 - \omega(G_{seq})}{d^2_a} = \omega(G_{par}).$
    \end{enumerate}
\end{lemma}

\begin{proof}
    Fix a strategy $S$ specified by $\rho_{AB}, \mc{E}_m, D_m$. Now consider two cases, first, where the measurement outcome of the verifier $V$ is not $\ket{\phi^{+}},$ and second, where the measurement outcome is $\ket{\phi^{+}}$. In the first case, the strategy succeeds with probability $1$, and this case occurs with probability $1 - \frac{1}{d^2_a}.$ For the second case, let $w$ be the winning probability conditioned on the verifier measuring $\ket{\phi^+}$. We will show that $w = \omega^{\ast}(G_{seq}; S)$. The winning probability of $G_{par}$ in this case is $w \coloneqq \expec_m d^2_a \Tr ((\ketbra{\phi^+}_{Aa} \otimes D^{Bb}_m) \mc{E}_m^A (\nu_{ab} \otimes \rho_{AB})).$ Here, the distribution is uniform over classical messages $m$. We know that $\nu_{ab} = (I \otimes \mc{N})(\ketbra{\phi^+}_{aa'})$ is the Choi state of $\mc{N}.$ To simplify this expression, first note that
    \begin{align*}
        d_a^2 \Tr_{Aa} (\ketbra{\phi^+}_{Aa} (\nu_{ab} \otimes \mc{E}_m^A(\rho_{AB}))) &= \Tr_{Aa} \sum_{i,i',j,j'} \ketbra{ii}{i'i'}_{Aa}(\ketbra{j}{j'}_{a}\otimes\mc{N}(\ketbra{j}{j'}_{b})\otimes\mc{E}_m^A(\rho_{AB})) \\
        &= \Tr_A \sum_{i,j} \ketbra{i}{j}_A (\mc{N}(\ketbra{j}{i}_b) \otimes \mc{E}_m^A(\rho_{AB})) \\
        &= (\mc{N} \otimes I_B)(\mc{E}_m^A(\rho_{AB})).
    \end{align*}
    Therefore, the winning probability $w = \expec_m \Tr(D_m^{Bb} (\mc{N} \otimes I_B)(\mc{E}_m^A(\rho_{AB}))) = \omega^{\ast}(G_{seq};S).$ Finally, we note that the second case occurs with probability $\frac{1}{d_a^2}.$ Combining the probabilities of the two cases we get that $1 - \frac{1}{d_a^2} \cdot 1 + \frac{1}{d_a^2} \cdot \omega^{\ast}(G_{seq};S) = \omega^{\ast}(G_{par};S)$ which gives the wanted result. For the second claim relating $\omega(G_{seq})$ with $\omega(G_{par})$ the same proof gives the result considering the shared state $\rho_{AB} = \rho_A \otimes \rho_B$ to be separable.
\end{proof}

\subsection{Zero-error quantum capacities}

In this section, we formalise a a quantum game and an interactive protocol for the zero-error quantum capacity and its entanglement-assisted counterpart.

\begin{definition}[Sequential 2-round protocol] \label{def:Hseq}
    Let $H_{seq}$ be a QMIP protocol where $V, A, B$ denote the systems of the verifier, and players Alice and Bob respectively. Let $\mc{N}$ be a noisy quantum channel. The game is played as follows:
    \begin{enumerate}
        \item $V$ samples $\psi$ from the uniform distribution on pure states in $\C^M$ and sends it to $A$.
        \item $A$ then outputs $\rho_A$, her share of a quantum state $\rho_{AB}$ potentially entangled with $B$.
        \item $V$ applies $\mc{N}: B(\mc{H}_a) \to B(\mc{H}_b)$ to $\rho_A$ and obtains $\rho_{B'}$ which he then sends to $B'$ as a question.
        \item $B$ receives $\rho_{B'}$ and then sends $\sigma \in \C^M$ as the answer to $V.$
        \item The players win if $\psi \otimes \sigma$ passes the swap test.
    \end{enumerate}

    The strategy needed to play this game is as follows:
    \begin{enumerate}
        \item $A$ and $B$ share an entangled state $\rho_{AB}.$ 
        \item $A$ has an encoding channel $\mc{E} : B(\mc{H}_A \otimes \C^M) \to B(\mc{H}_a).$
        \item $B$ performs a decoding channel $\mc{D} : B(\mc{H}_b \otimes \mc{H}_B) \to B(\C^M)$.
    \end{enumerate}
\end{definition}

\begin{theorem}[From channels to protocols] \label{thm:Game-channel-correspondence} Given a channel $\mc{N}$ there exists $\mathsf{QMIP}$ protocol $H_{seq}$ such that the following hold true:
    \begin{enumerate}
        \item $H_{seq}$ has a perfect entanglement-assisted strategy if and only if $Q_{0,E}^{(1)}(\mc{N}) \geq \log \abs*{M}.$
        \item $\omega(H_{seq}) = 1$ if and only if $Q_{0}^{(1)}(\mc{N}) \geq \log \abs*{M}.$
    \end{enumerate}
\end{theorem}

\begin{proof}
    We define the winning probability of $H_{seq}$ for a strategy $S$ following from the description of the game as $$\omega^{\ast}(H_{seq};S) \coloneqq \expec\limits_{\psi} \frac{1}{2} + \frac{1}{2} \braket{\psi}{\mc{D} \circ (\mc{N} \circ \mc{E}) \otimes I_B) (\psi \otimes \rho_{AB})}{\psi}.$$ 
    
    If $\omega^{\ast}(H_{seq};S) = 1,$ then $\expec\limits_{\psi} \frac{1}{2} + \frac{1}{2} \braket{\psi}{\mc{D} \circ (\mc{N} \circ \mc{E}) \otimes I_B) (\psi \otimes \rho_{AB})}{\psi} = 1$
    which means that $\braket{\psi}{\mc{D} \circ (\mc{N} \circ \mc{E}) \otimes I_B) (\psi \otimes \rho_{AB})}{\psi} = 1$ for each $\psi.$ This implies that $\mc{D} \circ (\mc{N} \circ \mc{E}) \otimes I_B) (\psi \otimes \rho_{AB}) = \psi.$ It then follows by definition that $Q_{0, E}^{(1)} (\mc{N}) \geq \log \abs*{M}.$ 
    
    For the converse, consider the communication protocol specified by encoding channel \\ $\mc{E} : B(\mc{H}_A \otimes \C^M) \to B(\mc{H}_a)$, the noisy channel $\mc{N}$ and a decoding channel $\mc{D} : B(\mc{H}_b) \otimes B(\mc{H}_B) \to B(\C^M).$ Let $Q^{(1)}_{0, E} \geq \log \abs*{M}.$ This means that for every $\psi$, $\mc{D} \circ (\mc{N} \circ \mc{E}) \otimes I_B) (\psi \otimes \rho_{AB}) = \psi.$ Therefore, for each $\psi$, $\braket{\psi}{\mc{D} \circ (\mc{N} \circ \mc{E}) \otimes I_B) (\psi \otimes \rho_{AB})}{\psi} = 1$. This implies \\ $\expec\limits_{\psi} \frac{1}{2} + \frac{1}{2} \braket{\psi}{\mc{D} \circ (\mc{N} \circ \mc{E}) \otimes I_B) (\psi \otimes \rho_{AB})}{\psi} = 1$.

    For the second claim with $\omega(H_{seq})$ the same proof gives the result considering the shared state $\rho_{AB} = \rho_A \otimes \rho_B$ to be separable.
\end{proof}

\begin{definition}[Parallel 1-round game] \label{def:game-qu}
    Let $H_{par}$ be a quantum game where $V, A, B$ denote the systems of the verifier, and players Alice and Bob respectively. Let $\mc{N}$ be a noisy quantum channel. The game is played as follows:
    \begin{enumerate}
        \item $V$ samples $\psi$ from the uniform distribution on pure states in $\C^M$ and sends it to $A$, and prepares $\nu_{ab}$, the Choi state of $\mc{N}$ and sends part $b$ to $B.$
        \item $A$ sends $\rho_A$ to $V$ which is her share of the entangled state $\rho_{AB}.$
        \item $B$ applies a channel, $\mc{D}$ to $Bb$ and sends the outcome $\sigma$ to $V.$
        \item $V$ measures in the Bell basis of $Aa$. If the outcome is not $\ket{\phi^{+}}$ then $V$ accepts. If the outcome is $\ket{\phi^{+}}$ then $V$ accepts if $\psi \otimes \sigma$ passes the swap test.
        \item The players win when $V$ accepts.
    \end{enumerate}

    The strategy needed to play this game is as follows:
    \begin{enumerate}
        \item $A$ and $B$ share an entangled state $\rho_{AB}.$ 
        \item $A$ has an encoding channel $\mc{E} : B(\mc{H}_A \otimes \C^M) \to B(\mc{H}_a).$
        \item $B$ performs a decoding channel $\mc{D} : B(\mc{H}_b \otimes \mc{H}_B) \to B(\C^M)$.
    \end{enumerate}
\end{definition}

\begin{lemma}[Multi-to-single round equivalence] \label{thm:ent-val-quant} Given a channel $\mc{N}$ there exists a quantum game $H_{par}$ such that the following hold true:
    \begin{enumerate}
        \item $1 - \frac{1 - \omega^{\ast}(H_{seq})}{d^2_a} = \omega^{\ast}(H_{par}).$
        \item $1 - \frac{1 - \omega(H_{seq})}{d^2_a} = \omega(H_{par}).$
    \end{enumerate}
\end{lemma}

\begin{proof}
    Fix a strategy $S$ specified by $\rho_{AB}, \mc{E}, \mc{D}$. Now consider two cases, first, where the measurement outcome of the verifier $V$ is not $\ket{\phi^{+}},$ and second, where the measurement outcome is $\ket{\phi^{+}}$. In the first case, the strategy succeeds with probability $1$, and this case occurs with probability $1 - \frac{1}{d^2_a}.$ For the second case, let $\mathfrak{w}$ be the winning probability conditioned on the verifier measuring $\ket{\phi^+}.$ We will show that $\mathfrak{w} = \omega^{\ast}(H_{seq};S)$. The winning probability of $H_{par}$ in this case is $\mathfrak{w} \coloneqq \expec\limits_{\psi} \frac{1}{2} + \frac{1}{2} \braket{\psi}{\mc{D} \circ (\mc{N} \circ \mc{E}) \otimes I_B) (\psi \otimes \rho_{AB})}{\psi}.$ Here, the distribution is uniform over pure states $\psi.$
    We know that $\nu_{ab} = (I \otimes \mc{N})(\ketbra{\phi^+}_{aa'})$ is the Choi state of $\mc{N}.$ To simplify this expression, first note that
    \begin{align*}
        &d_a^2 \Tr_{Aa} (\ketbra{\phi^+}_{Aa} (\nu_{ab} \otimes \mc{E}^A(\psi \otimes \rho_{AB}))) \\
        &= \Tr_{Aa} \sum_{i,i',j,j'} \ketbra{ii}{i'i'}_{Aa}(\ketbra{j}{j'}_{a}\otimes\mc{N}(\ketbra{j}{j'}_{b})\otimes\mc{E}^A(\psi \otimes \rho_{AB})) \\
        &= \Tr_A \sum_{i,j} \ketbra{i}{j}_A (\mc{N}(\ketbra{j}{i}_b) \otimes \mc{E}^A(\psi \otimes \rho_{AB})) \\
        &= (\mc{N} \otimes I_B)(\mc{E}^A(\psi \otimes \rho_{AB})).
    \end{align*}
    Therefore, the winning probability $\mathfrak{w} = \expec\limits_{\psi} \frac{1}{2} + \frac{1}{2} \braket{\psi}{\mc{D} \circ (\mc{N} \circ \mc{E}) \otimes I_B) (\psi \otimes \rho_{AB})}{\psi} = \omega^{\ast}(H_{seq};S).$ Finally, we note that the second case occurs with probability $\frac{1}{d_a^2}.$ Combining the probabilities of the two cases we get that $1 - \frac{1}{d_a^2} \cdot 1 + \frac{1}{d_a^2} \cdot \omega^{\ast}(H_{seq};S) = \omega^{\ast}(H_{par};S)$ which gives the wanted result.
    For the second claim relating $\omega(H_{seq})$ with $\omega(H_{par})$ the same proof gives the result considering the shared state $\rho_{AB} = \rho_A \otimes \rho_B$ to be separable.
\end{proof}

\section{From channel superactivation to interactive quantum protocols}

In this section, we will show that the quantum multiprover interactive protocols for the zero-error classical (\cref{thm:yipee}, \cref{cor:all-n-yay}) and quantum capacity (\cref{thm:phenom-quant}), $G_{seq}$ and $H_{seq}$ respectively, also admit the phenomenon of extreme violation of parallel repetition, i.e., \textit{superactivated cheating}. This is a multi-round analogue of what \cref{thm:foul-play} exhibits for a ($1$-round) quantum game. Superactivation of channel capacities allows the players to utilise the entanglement in the question state to win perfectly in $n$-instances of parallel repetition in spite of not winning perfectly in a single instance. This defies the traditional notion of parallel repetition for interactive proof systems or games and their local value.

\subsection{Cheating in the classical capacity QMIP protocol} \label{sec:cheat}

First, we note the correct setting for superactivation in zero-error classical capacities. We consider the unassisted zero-error classical capacity of a quantum channel, $C_{0} (\mc{N})$ in the asymptotic setting. The corresponding quantum protocol is still the same as $G_{seq}$ (\cref{def:cl-cap-2ro}) except now the state preshared by Alice and Bob is separable, $\rho_{AB} = \rho_A \otimes \rho_B$.

\begin{lemma}[Theorem 2, \cite{CS12}] \label{lem:superact}
    There exist channels $\mc{E}_1, ~\mc{E}_2$ with $C_0(\mc{E}_1) = C_0(\mc{E}_2) = 0$ such that $C_0(\mc{E}_1 \otimes \mc{E}_2) \geq 1.$
\end{lemma}

Since $C_0(\mc{E}) = \sup\limits_k \frac{C_0^{(1)}(\mc{E}^{\otimes k})}{k}$, it follows that $C_0(\mc{E}) = 0 $ if and only $ C_0^{(1)}(\mc{E}^{\otimes k}) = 0$ for all $k$. Therefore, we have the following corollary.

\begin{corollary} \label{cor:viol-par}
    Let $\mc{E}_1, \mc{E}_2$ be as in \cref{lem:superact}. For every positive integer $n$, there exists $k>1$ such that $C_0^{(1)}(\mc{E}_1^{\otimes k}) = C_0^{(1)}(\mc{E}_2^{\otimes k})=0$ and $C_0^{(1)}(\mc{E}^{\otimes k}_1 \otimes \mc{E}^{\otimes k}_2) \geq n$. 
\end{corollary}

The following lemma relates parallel repetition of $G_{seq}$ with running the protocol over the joint channel.

\begin{proposition} \label{thm:par-rep}
    Let $G_{seq}(\mc{N}, M)$ be the QMIP protocol for quantum channel $\mc{N}$ and number of messages $M.$ Then for any channels $\mc{N}_1, \mc{N}_2$ with messages $M_1, M_2$ respectively 
    \begin{equation}
        G_{seq}(\mc{N}_1, M_1) \times G_{seq}(\mc{N}_2, M_2) = G_{seq}(\mc{N}_1 \otimes \mc{N}_2, M_1 \times M_2).
    \end{equation}
\end{proposition}

\begin{proof}
    Consider the QMIP protocol $G_{seq}(\mc{N}_1, M_1) \times G_{seq}(\mc{N}_2, M_2)$. Let $m_1, m_2$ be messages received by $A$ from $V$ and let $\rho_{A_1A_2}$ be the output state of $A$ that is returned to $V$. $V$ acts with $\mc{N}_1 \otimes \mc{N}_2 (\rho_{A_1A_2}) = \rho_{B_1B_2}$ which is the question sent to $B.$ $B$ after receiving $\rho_{B_1B_2}$ sends answers $m_1', m_2'$ to $V.$ $A$ and $B$ win if and only if $m_1 = m_1'$ and $m_2 = m_2'.$ The above QMIP protocol describes the exact protocol $G_{seq}(\mc{N}_1 \otimes \mc{N}_2, M_1 \times M_2)$ where $V$ acts with the joint channel $\mc{N}_1 \otimes \mc{N}_2$ and the set of messages $[M_1 M_2]$ is in bijection with the joint set $[M_1] \times [M_2].$
\end{proof}

An immediate corollary shows that parallel repetition does not hold for the local value of two distinct $\mathsf{QMIP}$ protocols. This is a stepping stone in the direction which ultimately shows the extreme violation of parallel repetition for the local value of the same protocol.

\begin{corollary}
    $\exists$ QMIP protocols $G_1, G_2$ with $\omega(G_1) < 1$ and $\omega(G_2) < 1$ but $\omega(G_1 \times G_2) = 1.$
\end{corollary}

\begin{proof}
    Let $\mc{E}_1, \mc{E}_2, k$ be as in \cref{cor:viol-par}. Now let $G_1 = G_{seq}(\mc{E}_1^{\otimes k}, 2)$ and $G_2 = G_{seq}(\mc{E}_2^{\otimes k}, 2)$. Since $C_0^{(1)}(\mc{E}_1^{\otimes k}) = 0$ and $C_0^{(1)}(\mc{E}_2^{\otimes k}) = 0$, $\omega(G_1) < 1$ and $\omega(G_2) < 1.$ But by \cref{cor:viol-par} we know $C_0^{(1)}(\mc{E}_1^{\otimes k} \otimes \mc{E}_2^{\otimes k}) \geq 2.$ This implies $\omega(G_{seq}(\mc{E}_1^{\otimes k} \otimes \mc{E}_2^{\otimes k}, 4)) = 1$. Now applying this fact to \cref{thm:par-rep} we get, $\omega(G_1 \times G_2) = 1.$
\end{proof}

We are now going to describe a technical invention that will helps us establish our main result. Intuitively, this is a direct sum construction for channels such that any possible correlations between two channel outputs is broken by the direct sum of the two channels, which we call the \emph{one time pad} direct sum. This will enable us to establish that a failure of perfect execution in one instance is mitigated by superactivation of the same resources when interacting over multiple rounds leading to the failure of parallel repetition.

\begin{definition}[One-time pad direct sum] \label{def:OTP-dir-sum}
    Let $\Phi_0, \Phi_1$ be channels with a common output space. We define $\Phi_0 \oplus_{\mathrm{OTP}} \Phi_1$ as the channel that acts on a state
    $$ \rho = \begin{pmatrix} \rho_{00} & \rho_{01} \\
    \rho_{10} & \rho_{11} \end{pmatrix}$$
    as
    \begin{align*}
        \Phi_0 \oplus_{\mathrm{OTP}} \Phi_1 (\rho)
        = \frac{1}{\abs{K}} \sum\limits_k \left( \Phi_0 (\rho_{00}) + U_k \Phi_1 (\rho_{11}) U_k^{\dag} \right) \otimes \ketbra{k}
    \end{align*}
    where $\{U_k\}_{k \in K}$ is a unitary $1$-design.    
\end{definition}

\begin{remark}\label{rmk:1-design}
    Recall that for unitary $1$-designs, $\frac{1}{\abs{K}}\sum_k U_k X U_k^\dagger = (\Tr(X)/d) I$. Moreover, if the channels $\Phi_0, \Phi_1$ in \cref{def:OTP-dir-sum} do not have a common output space, then it is always to embed the smaller output space into a larger dimension to obtain a common output space.
\end{remark}

We now note useful properties of the one-time-pad direct sum, which will later help prove our main result.

\begin{proposition}\label{lem:OTP-pos-cap}
    Let $\mc{E}_0:B(\mc{H}_0)\rightarrow B(\mc{K})$ and $\mc{E}_1:B(\mc{H}_1)\rightarrow B(\mc{K})$ be two quantum channels with a common output space, and let $\mc{N}=\mc{E}_0 \oplus_{\mathrm{OTP}} \mc{E}_1$. Then for 
every $n\geq 1$,
        \begin{equation}
            C^{(1)}_0(\mc{N}^{\otimes n})\geq \max\{C^{(1)}_0(\mc{E}_{s_1}\otimes\cdots\otimes\mc{E}_{s_n}): s\in\{0,1\}^n\}. \label{eq:n-geq}
        \end{equation}
Moreover, for $n=1$,
        \begin{equation}
            C^{(1)}_0(\mc{N})=\max\{C^{(1)}_0(\mc{E}_0),C^{(1)}_0(\mc{E}_1)\}. \label{eq:1-eq}
        \end{equation}
\end{proposition}

\begin{proof}
Fix an arbitrary $s\in\{0,1\}^n$. Let $\mc{H}_s:=\mc{H}_{s_1}\otimes\cdots\otimes\mc{H}_{s_n}$, and let $\mc{E}_s:=\mc{E}_{s_1}\otimes \cdots\otimes\mc{E}_{s_n}$. Let $V_i$ be the canonical isometry from $\mc{H}_i\rightarrow \mc{H}_0\oplus\mc{H}_1$, so $V_s:=V_{s_1}\otimes\cdots\otimes V_{s_n}$ is a isometry from $\mc{H}_s\rightarrow (\mc{H}_0\oplus\mc{H}_1)^{\otimes n}$. 

Suppose $C_0^{(1)}(\mc{E}_s)=\log(M)$. This means that there exist densities $\rho_1,\ldots,\rho_M$ on $\mc{H}_s$ such that $\Tr\big(\mc{E}_s(\rho_i)\mc{E}_s(\rho_j)  \big)=0$ for all $i\neq j$. For every $1\leq i\leq M$, let $\hat{\rho}_i:= V_s\rho_i V_s^\dagger$ be a density operator on $(\mc{H}_0\oplus\mc{H}_1)^{\otimes n}$. Let $\{U_k\}_{k\in K}$ be a unitary 1-design on $\mc{K}$. From \Cref{def:OTP-dir-sum}, we see that for any density $\sigma$ on $\mc{H}_0$, 
\begin{equation*}
    \mc{N}(V_0\sigma V_0^*)= \frac{1}{\abs{K}}\sum_{k\in K} \mc{E}_0(\sigma)\otimes \ketbra{k}{k},
\end{equation*}
and for any density $\sigma$ on $\mc{H}_1$,
\begin{equation*}
     \mc{N}(V_1\sigma V_1^*)= \frac{1}{\abs{K}}\sum_{k\in K} U_k\mc{E}_1(\sigma)U_k^\dagger\otimes \ketbra{k}{k}.
\end{equation*}
For any $k=k_1\cdots k_n\in K^n$, define the unitary $U_{s,k}:=U_{k_1}^{s_i}\otimes\cdots\otimes U_{k_n}^{s_n}$, where $U_{k_i}^{s_i}=\begin{cases}
    I_{\mc{K}} & \text{ if } s_i=0\\
    U_{k_i} & \text{ if } s_i=1.
\end{cases}$. It follows that
\begin{equation*}
    \mc{N}^{\otimes n} (\hat{\rho}_i) = \frac{1}{\abs{K}^n}\sum_{k\in K^n}U_{s,k}\mc{E}_s(\rho_i)U_{s,k}^\dagger \otimes\ketbra{k}{k}
\end{equation*}
for every $1\leq i\leq M$. Hence for every $i\neq j$,
\begin{align*}
    \Tr\big(  \mc{N}^{\otimes n} (\hat{\rho}_i)  \mc{N}^{\otimes n} (\hat{\rho}_j)     \big)&=\frac{1}{\abs{K}^{2n}} \sum_{k\in K^n} \Tr\big(U_{s,k}\mc{E}_s(\rho_i)U_{s,k}^\dagger U_{s,k}\mc{E}_s(\rho_j)U_{s,k}^\dagger   \big) \\
    & = \frac{1}{\abs{K}^{2n}} \sum_{k\in K^n} \Tr\big(\mc{E}_s(\rho_i)\mc{E}_s(\rho_j)\big)=0.
\end{align*}
We conclude that $C^{(1)}_0(\mc{N}^{\otimes n})\geq \log(M)$. This proves \Cref{eq:n-geq} and the ``$\geq$" direction of \Cref{eq:1-eq}.

Now we only need to prove the ``$\leq$" direction of \Cref{eq:1-eq}. Suppose $C_0^{(1)}(\mc{N})=\log(M)$. Then there exist pairwise orthogonal densities $\sigma^x=\begin{pmatrix}
    \sigma^x_{00} & \sigma^x_{01}\\
    \sigma^x_{10} & \sigma^x_{11}
\end{pmatrix}$, $1\leq x\leq M$ on $\mc{H}_0\oplus \mc{H}_1$. For the notational convenience, we let $A_x:=\mc{E}_0(\sigma^x_{00}), B_x:=\mc{E}_1(\sigma^x_{11})$, and  $a_x:= \Tr(\sigma^x_{00}),b_x:=\Tr(\sigma^x_{11})$ for every $x$. Then for every $x\neq y$,
\begin{align*}
    0&=\Tr\big(\mc{N}(\sigma^x) \mc{N}(\sigma^y)  \big) = \frac{1}{\abs{K}^2}\sum_{k\in K}\Tr\big((A_x+U_kB_xU_k^\dagger)(A_y+U_yB_yU_k^\dagger) \big)\\
    &= \frac{1}{\abs{K}}\Tr(A_xA_y) + \frac{1}{\abs{K}}\Tr(B_xB_y) + \frac{1}{\abs{K}^2}\Tr\big( A_x (\sum_{k\in K} U_kB_yU_k^\dagger) + (\sum_{k\in K} U_kB_xU_k^\dagger)A_y  \big)  \\
    & = \frac{1}{\abs{K}}\Tr(A_xA_y) + \frac{1}{\abs{K}}\Tr(B_xB_y) + \frac{1}{\dim(\mc{K})\abs{K}} a_xb_y+\frac{1}{\dim(\mc{K})\abs{K}}a_yb_x,
\end{align*}
where the last equality uses \Cref{rmk:1-design}. Since all the above four terms are nonnegative, we must have 
\begin{equation}
    a_xb_y=a_yb_x=0\label{eq:==}
\end{equation}
for all $x\neq y$. Note that $a_x+b_x=1$ for every $x$. Assume for some $y$, both $a_y$ and $b_y$ are $>0$. Then for any $x\neq y$, \Cref{eq:==} implies that $a_x=b_x=0$, a contradiction. We conclude that for every $x$, either $a_x=0,b_x=1$ or $a_x=1,b_x=0$. Now fix an $x$, if $a_x=0$ and $b_x=1$, then \Cref{eq:==} implies $a_y=0$ and $b_y=1$ for all $y\neq x$. Similarly, if $a_x=1$ and $b_x=0$, then we must have $a_y=1$ and $b_y=0$ for all $y\neq x$. We conclude that, either $a_x=0,b_x=1$ for all $x$, in which case every $\sigma^x$ is supported on $\mc{H}_1$ and hence $C_0^{(1)}(\mc{E}_1)\geq \log(M)$, or $a_x=1,b_x=0$ for all $x$, in which case every $\sigma^x$ is supported on $\mc{H}_0$ and hence $C_0^{(1)}(\mc{E}_0)\geq \log(M)$. This proves the ``$\leq$" direction of \Cref{eq:1-eq}.
\end{proof}

Finally, we state and prove our main result which establishes an extreme failure of parallel repetition for $\mathsf{QMIP}$ protocols with multiple nontrivial rounds.

\begin{theorem} \label{thm:yipee}
    There exists a QMIP protocol $G$ with $\omega(G) < 1$ but $\omega(G^3) = \omega(G^2) = 1.$
\end{theorem}
\begin{proof}
    Let $\mc{E}_1, \mc{E}_2$ be as in \cref{lem:superact}. Then by \Cref{cor:viol-par}, there exists $k>1$ such that 
    \begin{equation*}
        C_0^{(1)}(\mc{E}_1^{\otimes k})=C_0^{(1)}(\mc{E}_2^{\otimes k})=0 \text{  and  } C_0^{(1)}(\mc{E}_1^{\otimes k}\otimes \mc{E}_2^{\otimes k})\geq 3. 
    \end{equation*}
    For the notational convenience, let $\mc{F}_i:=\mc{E}_i^{\otimes k}$ for $i=1,2$, and let $\mc{N}:=\mc{F}_1\oplus_{\mathrm{OTP}}\mc{F}_2$. 
    Let $G = G_{seq}(\mc{N}, 2)$. By \Cref{eq:1-eq} in  \Cref{lem:OTP-pos-cap}, 
    \begin{equation*}
        C_0^{(1)}(\mc{N})=\max\{C_0^{(1)}(\mc{F}_1),C_0^{(1)}(\mc{F}_2)\} =0.
    \end{equation*}
    Then by \Cref{thm:Game-channel-correspondence}, we must have  $\omega(G) < 1$. Applying \Cref{eq:n-geq} of \Cref{lem:OTP-pos-cap} with $n=2$ yields
    \begin{equation*}
        C_0^{(1)}(\mc{N}^{\otimes 2})\geq C_0^{(1)}(\mc{F}_1\otimes\mc{F}_2)\geq 3. \label{eq:geq3}
    \end{equation*}
    So in particular, 
    \begin{equation*}
        C_0^{(1)}(\mc{N}^{\otimes 2})\geq 2=\log(4),
    \end{equation*}
    and 
    \begin{equation*}
        C_0^{(1)}(\mc{N}^{\otimes 3})\geq C_0^{(1)}(\mc{N}^{\otimes 2})\geq 3=\log(8),
    \end{equation*}
    where the first inequality uses the fact that $C_0^{(1)}(\Phi\otimes\Psi)\geq C_0^{(1)}(\Phi)$ for all channels $\Phi$ and $\Psi$. \Cref{thm:par-rep} implies that $G^2=G_{seq}(\mc{N}^{\otimes 2}, 4)$ and $G^3=G_{seq}(\mc{N}^{\otimes 3}, 8)$.
    Then it follows again from \Cref{thm:Game-channel-correspondence} that $\omega(G^2)=\omega(G^3)=1$.
\end{proof}

\begin{corollary} \label{cor:all-n-yay}
    $\exists$ QMIP protocol $G$ with $\omega(G) < 1$ but $\omega(G^n) = 1$ for all $n \geq 2.$
\end{corollary}

\subsection{Cheating in the quantum capacity QMIP protocol} 

In this section we show that the interactive protocol capturing the zero-error quantum capacity of a channel exhibits the same cheating phenomenon as its classical counterpart via superactivation. Once again, the setting of superactivation of channel capacities concerns the unassisted zero-error quantum capacity of a quantum channel, $Q_0(\mc{N})$ in the asymptotic setting. Precisely, we show that $H_{seq}$ (see \cref{def:Hseq}) is also a class of interactive protocols, for which $\omega(H_{seq}) < 1$, but $\omega(H_{seq}^n) = 1,$ for all $n \geq 2,$ and in the present context, the state prehsared by Alice and Bob is separable, $\rho_{AB} = \rho_A \otimes \rho_B.$ The proof closely follows the case of the classical capacity protocol, hence we outline a sketch with the parts that are distinctive. 

The following lemma shows that the asymptotic zero-error quantum capacity superactivates.

\begin{lemma}[Theorem 3, \cite{CS12}] \label{lem:qua-super}
    $\exists$ channels $\mc{E}_1, ~\mc{E}_2$ with $Q_0(\mc{E}_1) = Q_0(\mc{E}_2) = 0$ such that $Q_0(\mc{E}_1 \otimes \mc{E}_2) \geq 1.$
\end{lemma} 

Since $Q_0(\mc{E}) = \sup\limits_k \frac{Q_0^{(1)}(\mc{E}^{\otimes k})}{k}$, it follows that $Q_0(\mc{E}) = 0$ if and only if $Q_0^{(1)}(\mc{E}^{\otimes k}) = 0, ~\forall k.$ Therefore, we have the following corollary.

\begin{corollary} \label{cor:quan-viol}
    Let $\mc{E}_1, \mc{E}_2$ be as in \cref{lem:qua-super}. For every positive integer $n$, there exists $k>1$ such that $Q_0^{(1)}(\mc{E}_1^{\otimes k}) = Q_0^{(1)}(\mc{E}_2^{\otimes k})=0$ and $Q_0^{(1)}(\mc{E}^{\otimes k}_1 \otimes \mc{E}^{\otimes k}_2) \geq n$.  
\end{corollary}

Next, we obtain the following result, the proof of which follows the same argument as that for the classical capacity case, starting with the zero-error quantum capacity analogue of \cref{thm:par-rep}, and onwards in \cref{sec:cheat}.

This result shows a different class of $\mathsf{QMIP}$ protocols corresponding to the zero-error quantum capacity of a channel, whose local value shows extreme failure of parallel repetition.

\begin{theorem} \label{thm:phenom-quant}
    There exists a class of QMIP protocols $H_{seq}$ with $\omega(H_{seq}) < 1$ but $\omega(H_{seq}^n) = 1$ for all $n \geq 2.$
\end{theorem}

\section{Impossibility of superactivation}

In this section, we show instances where superactivated cheating is impossible; one with particular cryptographic relevance (\cref{thm:mult-imposs}), and others with consequences for quantum channel capacities (\cref{thm:C-imp-sup} and \cref{cor:Q-imp-sup}).

\subsection{Impossibility of superactivation in channel capacities}

First, we show that as a consequence of our impossibility result — the entangled value of a quantum game cannot be superactivated (\cref{thm:imposs-q-game}), certain entanglement-assisted zero-error capacities cannot be superactivated. This follows from the more general idea that for any channel capacity, impossibility of $1$-shot superactivation implies that for the asymptotic case.

\begin{lemma}[$1$-shot impossibility of superactivation extends to asymptotic case] \label{thm:1-shot-superactivation-to-asymptotic}
    For any channel capacity $T(\mc{N})$, if $T^{(1)}(\mc{N})$ cannot be superactivated, then $T(\mc{N})$ cannot be superactivated, where $T(\mc{N}) = \lim\limits_{n \to \infty} \frac{T^{(1)}(\mc{N}^{\otimes n})}{n}.$ 
\end{lemma}

\begin{proof}
    Since $T^{(1)}(\mc{N})$ cannot be superactivated, it implies that if $T^{(1)}(\mc{N}) = 0$, then $T^{(1)}(\mc{N}^{\otimes 2}) = 0.$ More generally, for channels $\mc{M}, \mc{N}$, if $T^{(1)}(\mc{M}) = T^{(1)}(\mc{N}) = 0$, then~$T^{(1)}(\mc{M} \otimes \mc{N}) = 0.$ We proceed with the case where $\mc{M} = \mc{N}$ to establish the claim for the asymptotic capacity which is a regularisation of the one-shot capacity, i.e., $T(\mc{N}) = \lim\limits_{n \to \infty} \frac{T^{(1)}(\mc{N}^{\otimes n})}{n}.$ Take $\mc{N}^{\otimes 2} = \mc{N}'.$ From hypothesis, $T^{(1)}(\mc{N}) = 0 \implies T^{(1)}(\mc{N}') = 0.$ Therefore, $T^{(1)}(\mc{N} \otimes \mc{N}') = 0 = T^{(1)}(\mc{N}^{\otimes 3}).$ Since we have shown that $T^{(1)}(\mc{N}) = 0 \implies T^{(1)}(\mc{N}^{\otimes n}) = 0$ for $n = 2, 3$ this implies that this fact generally holds for all $n.$ Therefore in the regularised formula, $T^{(1)}(\mc{N}^{\otimes n}) = 0$ for all $n$. This implies that $T(\mc{N}) = \lim\limits_{n \to \infty} \frac{T^{(1)}(\mc{N}^{\otimes n})}{n} = 0$ if $T^{(1)}(\mc{N}) = 0$ completing the proof.
\end{proof}

Using this, we show that the entanglement-assisted zero-error classical capacity (both one-shot and asymptotic) cannot be superactivated.

\begin{theorem}[Impossibility of superactivation for the classical case] \label{thm:C-imp-sup}
    For any channels $\mc{N}_1$ and $\mc{N}_2$, if $C_{0, E}^{(1)}(\mc{N}_1) = C_{0, E}^{(1)}(\mc{N}_2) = 0$, then $C_{0, E}^{(1)}(\mc{N}_1 \otimes \mc{N}_2) = 0$. Consequently, $C_{0, E}(\mc{N})$ cannot be superactivated.
\end{theorem}

\begin{proof}
    Let $G_{par}^1$ and $G_{par}^2$ be quantum games for the zero-error classical capacity as defined in \cref{def:game-clas} corresponding to channels $\mc{N}_1$ and $\mc{N}_2$. By \cref{thm:imposs-q-game}, we have that $\omega^{\ast}(G_{par}^1 \times G_{par}^2) \leq \omega^{\ast}(G_{par}^i)$ where $i = 1, 2.$ Combining this with \cref{thm:game-prot-val} the claim holds. The last claim follows from the first via \cref{thm:1-shot-superactivation-to-asymptotic}.
\end{proof}

Now we note that, the quantum zero-error entanglement-assisted capacity (asymptotic) is related to the analogous classical capacity by the following formula

\begin{equation}\label{eq:rel-C-Q}
    \frac{1}{2} C_{0,E}(\mc{N}) = Q_{0,E}(\mc{N}).
\end{equation}

The above relation from the duality of teleportation and superdense coding, and analogous arguments for the finite-error capacities as given by the reverse Shannon theorem \cite{BSST02}. Via this relation, we will now show that $Q_{0,E}(\mc{N})$ cannot be superactivated in stark contrast to its unassisted counterpart, and to the best of our knowledge this was not known earlier.

\begin{corollary}[Impossibility of superactivation for the asymptotic quantum case] \label{cor:Q-imp-sup}
    For any channels $\mc{N}_1$ and  $\mc{N}_2$, if $Q_{0, E}(\mc{N}_1) = Q_{0, E}(\mc{N}_2) = 0$, then $Q_{0, E}(\mc{N}_1 \otimes \mc{N}_2) = 0$.
\end{corollary}

\begin{proof}
    Since $Q_{0,E}(\mc{N}_1) = Q_{0,E}(\mc{N}_2) = 0$, then by \cref{eq:rel-C-Q} $C_{0,E}(\mc{N}_1) = C_{0,E}(\mc{N}_2) = 0.$ Now by \cref{thm:C-imp-sup}, and the relation in \cref{eq:rel-C-Q}, the claim follows.
\end{proof}

\begin{remark}
Interestingly, the relation between the one-shot cases is more subtle and given by
\begin{equation}\label{eq:1-shot-C-Q}
Q_{0, E}^{(1)}(\mc{N}) = \log \floor{2^{C_{0,E}^{(1)}(\mc{N})/2}}.
\end{equation}
To see this, note that a perfect $d$-dimensional quantum code combined with superdense coding and an additional maximally entangled state transmits $d^2$ classical messages, and hence the maximum number of entanglement-assisted zero-error classical messages is lower bounded by $d^2.$ Now, a classical code containing at least $d^2$ messages transmits the $d^2$ possible teleportation keys, and hence transmits a $d$-dimensional quantum state. We also note that the general idea behind the relation still originates from the reverse Shannon theorem.

\begin{remark}
    Incidentally, $Q^{(1)}_{0,E}$, the one-shot entanglement-assisted zero-error quantum capacity can be superactivated. Consider the completely dephasing channel denoted by $\Delta$, whose action is given by $\Delta(\rho) = \sum\limits_{x = 0}^1 \braket{x}{\rho}{x} \ketbra{x}.$ This is a noiseless one-bit classical channel. By \cite{DSW13} $C_{0,E}^{(1)}(\Delta) = 1$, and hence by \cref{eq:1-shot-C-Q}, $Q_{0,E}^{(1)}(\Delta) = 0.$ Now, $2$ uses of $\Delta$ transmit $2$ classical bits perfectly. Combined with a preshared maximally entangled state, the $2$ bits can perform teleporation to achieve $Q_{0,E}^{(1)}(\Delta^{\otimes 2}) \geq 1.$ Therefore, $Q_{0,E}^{(1)}(\Delta) = 0,$ but $Q_{0,E}^{(1)}(\Delta^{\otimes 2}) \geq 1$ showing that superactivation is possible.
\end{remark}

\end{remark}

\subsection{Impossibility of superactivated cheating in multi-round protocols}

Here, we will show that superactivation cannot happen for the entangled-value of multi-round games or QMIP protocols when the number of messages exchanged is small. The particular class of protocols for which our impossibility result holds is cryptographically relevant as it models a security game for various primitives such as commitments and interactive arguments (see more in, for example, \cite{BQSY24}). First, we define the QMIP protocol.

\begin{definition}[3-message QMIP protocol] \label{def:3-mes}
    Let $K$ be a QMIP protocol with $V, A, B$ denoting the verifier and the two players Alice and Bob. Let $\rho_{AB}$ be an entangled state preshared by $A$ and $B.$ Let $\mc{N}$ be a quantum channel. The game is played as follows:
    \begin{enumerate}
        \item $A$ and $B$ send registers $A_1$ and $B_1$ to $V.$
        \item $V$ applies $\mc{N}$ to $A_1$ and $B_1$ and sends the respective outputs $A_2$ and $B_2$ to the players. $V$ keeps the reference register $R$.
        \item $A$ and $B$ apply channels $\mc{E}_A$ and $\mc{E}_B$ to $A_2$ and $B_2$ respectively to obtain outputs $A_3$ and $B_3.$
        \item The players send $A_3$ and $B_3$ to $V$.
        \item $V$ performs the two-outcome measurement $\{\Pi_{\mathrm{acc}}, I - \Pi_{\mathrm{acc}}\}$ on the joint system $A_3B_3R$. The players win if the outcome corresponding to $\Pi_{\mathrm{acc}}$ occurs.
    \end{enumerate}
    An entangled strategy includes the shared state $\rho_{AB}$ and channels $\mc{E}_A$ and $\mc{E}_B$ while the local strategy only includes the channels. We denote by $\omega(K),$ the local value and by $\omega^{\ast}(K)$ the entangled value.
\end{definition}

\begin{theorem}[Impossibility of superactivated cheating for entangled value of QMIP protocol] \label{thm:mult-imposs}
    For all 3-message QMIP protocols $K, K'$ as in \cref{def:3-mes}, not necessarily distinct, $$\omega^{\ast}(K \times K') \leq \omega^{\ast}(K).$$    
\end{theorem}

\begin{proof}
    Consider a strategy for $K \times K'$ consisting of $(\rho_{AA'BB'}, \mc{E}_{AA'}, \mc{E}_{BB'}).$ To construct a strategy for the single instance $K$, the players could choose their shared state to be $\rho_{AA_2'BB_2'} = \mc{N}'(\rho_{AA'BB'}).$ To complete the strategy, they could take their respective channels to be $\mc{E}_{AA'}$ and $\mc{E}_{BB'}.$ Now when the verifier applies $\mc{N}$, the shared state becomes $\rho_{A_2A_2'B_2B_2'}.$ This is the same as the shared state after the first round in the protocol $K \times K',$ which is when $A$ and $B$ apply channels $\mc{E}_{AA'}$ and $\mc{E}_{BB'}.$ After the final message, $V$ holds $A_3B_3R.$ 
    This is the marginal corresponding to $K$ of the same state that $V$ would hold at the end of the protocol $K \times K'$. The probability that $V$ accepts using $\Pi_{\mathrm{acc}}$ on this state is an upper bound on the probability of $V$ accepting using $\Pi_{\mathrm{acc}} \otimes \Pi_{\mathrm{acc}}'$ on this state at the end of the protocol $K \times K'.$ Since this holds for all strategies, it completes the proof.
\end{proof}

\bibliographystyle{bibtex/bst/alphaarxiv.bst}
\bibliography{bibtex/bib/quasar-full.bib,
              bibtex/bib/quasar.bib,
              bibtex/bib/quasar-more.bib}
              
\end{document}